\documentclass[11pt]{article}
\usepackage{times}
\usepackage{xspace,comment,calc}
\usepackage{amssymb,amsmath,amsthm,amsfonts,amstext}
\usepackage[margin=1cm]{caption}
\usepackage[shortlabels]{enumitem}
\usepackage{multirow}
\usepackage{graphicx}
\usepackage{color}
\usepackage{bbm}

\newtheorem{theorem}{Theorem}[section]
\newtheorem{lemma}[theorem]{Lemma}
\newtheorem{claim}[theorem]{Claim}

\newtheorem{corollary}[theorem]{Corollary}

{\theoremstyle{definition}  }

\newenvironment{proofof}[1]{\begin{proof}[Proof of #1]}{\end{proof}}

\newcommand{\bg}[1]{\medskip\noindent{\it #1}}

\newcommand{\np}{{\em NP}\xspace}
\newcommand{\nphard}{\np-hard\xspace}

\newcommand{\R}{\ensuremath{\mathbb R}}

\newcommand{\A}{\ensuremath{\mathcal{A}}}
\newcommand{\B}{\ensuremath{\mathcal{B}}}

\newcommand{\Oc}{\ensuremath{\mathcal O}}

\newcommand{\OPT}{\ensuremath{\mathit{OPT}}}

\newcommand{\frall}{\ensuremath{\text{ for all }}}

\newcommand{\sm}{\ensuremath{\setminus}}
\newcommand{\es}{\ensuremath{\emptyset}}

\newcommand\restr[2]{{% we make the whole thing an ordinary symbol
  \left.\kern-\nulldelimiterspace % automatically resize the bar with \right
  #1 % the function
  \vphantom{\big|} % pretend it's a little taller at normal size
  \right|_{#2} % this is the delimiter
  }}

\newcommand{\e}{\ensuremath{\epsilon}}

\newcommand{\sse}{\subseteq}

\newcommand{\tx}{\ensuremath{\tilde x}}

\newcommand{\hx}{\ensuremath{\hat x}}

\newcommand{\bx}{\ensuremath{\bar x}}

\newcommand{\ld}{\ensuremath{\lambda}}

\newcommand{\al}{\ensuremath{\alpha}}
\newcommand{\tht}{\ensuremath{\theta}}

\newcommand{\dt}{\ensuremath{\delta}}

\newcommand{\Gm}{\ensuremath{\Gamma}}
\newcommand{\sg}{\ensuremath{\sigma}}

\newcommand{\val}{\ensuremath{\mathrm{val}}}

\newcommand{\child}{\ensuremath{\mathsf{ch}}}
\newcommand{\prt}{\ensuremath{\mathcal{S}}}
\newcommand{\level}{\ensuremath{L}}

\newcommand{\dks}{\ensuremath{\mathsf{D}k\mathsf{S}}\xspace}
\newcommand{\mindks}{\ensuremath{\mathsf{Min}}\dks}
\newcommand{\maxst}{\ensuremath{\mathsf{MaxST}}\xspace}
\newcommand{\new}{\ensuremath{\mathsf{new}}}
\newcommand{\bgd}{\ensuremath{\mathsf{BGD}}\xspace}
\newcommand{\pok}{\ensuremath{\mathsf{POK}}\xspace}

\newcommand{\opttknlp}{\ensuremath{\OPT_{\text{\ref{tknlp}}}}}

\newcommand{\UB}{\ensuremath{\mathit{UB}}}
\DeclareMathOperator{\rk}{rk}

\title{Improved Algorithms for MST and Metric-TSP Interdiction%
\footnote{A preliminary version~\cite{LinharesS17} appeared in the Proceedings of the
  International Colloquium on Automata, Languages and Programming (ICALP), 2017.}}

\author{
         Andr\'e Linhares\thanks{{\tt \{alinhare,cswamy\}@uwaterloo.ca}. 
         Dept. of Combinatorics and Optimization, Univ. Waterloo, Waterloo, ON N2L 3G1.
         Supported in part by NSERC grant 327620-09 and an NSERC Discovery Accelerator
         Supplement award.}   
\and
\addtocounter{footnote}{-1}
         Chaitanya Swamy\footnotemark
}

\date{}

\begin{document}

\maketitle

\begin{abstract}
We consider the {\em MST-interdiction} problem: given a multigraph $G = (V, E)$, edge
weights $\{w_e\geq 0\}_{e \in E}$, interdiction costs $\{c_e\geq 0\}_{e \in E}$, and an
interdiction budget $B\geq 0$, the goal is to remove a set $R\sse E$ of edges of
total interdiction cost at most $B$ so as to maximize the $w$-weight of an MST of
$G-R:=(V,E\sm R)$. 

Our main result is a $4$-approximation algorithm for this problem. 
This improves upon the previous-best
$14$-approximation~\cite{Zenklusen15}. 
Notably, our analysis is also significantly simpler 
and cleaner than the one in~\cite{Zenklusen15}. 
Whereas~\cite{Zenklusen15} uses a greedy algorithm with an involved
analysis to extract a good interdiction set from an over-budget set, 
we utilize a generalization of knapsack called the {\em tree knapsack problem} that nicely
captures the key combinatorial aspects of this ``extraction problem.'' 
We prove a simple, yet strong, LP-relative approximation bound for tree knapsack, which 
leads to our improved guarantees for MST interdiction. Our algorithm and analysis are
nearly tight, as we show that one cannot achieve an approximation ratio better than 3
relative to the upper bound used in our analysis (and the one in~\cite{Zenklusen15}).

Our guarantee for MST-interdiction yields an $8$-approximation for 
{\em metric-TSP interdiction} (improving over the $28$-approximation
in~\cite{Zenklusen15}). We also show that 
the {\em maximum-spanning-tree interdiction} problem is at least as hard to approximate as
the minimization version of densest-$k$-subgraph.
\end{abstract}

\section{Introduction} 
Interdiction problems are a broad class of optimization problems with a wide range of
applications. They model the
problem faced by an attacker, who given an underlying, say, minimization, problem, 
aims to  
destroy or {\em interdict} the elements involved in the optimization problem (e.g., nodes
or edges in a network-optimization problem) 
without exceeding a given interdiction budget, so as to maximize the optimal value of the
residual optimization problem (where one cannot use the interdicted elements). 
A classical example is the {\em minimum-spanning-tree (MST) interdiction}
problem~\cite{LinC93,FredericksonS99,Zenklusen15}, which is the focus of this work: 
we are given a multigraph $G = (V, E)$, edge weights $\{w_e\geq 0\}_{e \in E}$,
interdiction costs $\{c_e\geq 0\}_{e \in E}$, and an interdiction budget 
$B\geq 0$; the goal is to interdict (i.e., remove) a set $R\sse E$ of edges of
total interdiction cost at most $B$ so as to maximize the $w$-weight of an MST of the
multigraph $G-R:=(V,E\sm R)$. 
Note that $G$ may have parallel
edges, which 
can be useful in modeling partial-interdiction effects, wherein interdicting 
an edge causes an increase in its weight that depends on the interdiction cost incurred
for the edge.  

At a high level, interdiction problems can be seen as investigating the sensitivity of an
underlying optimization problem with respect to the removal of a limited set of underlying
elements. This type of sensitivity analysis may be utilized to identify vulnerable
spots (e.g., regions in a network) either: 
(a) for possible reinforcement, or, (b) if the optimization problem models 
an undesirable process (e.g., the spread of infection, or nuclear-arms smuggling),
for disruption, so as to maximally impair the underlying process.
A variety of applications of interdiction problems ensue from these
two perspectives, including infrastructure protection~\cite{ChurchSM04,SalmeronWB09},
hospital-infection control~\cite{Assimakopoulos87}, 
prevention of nuclear-arms smuggling~\cite{PanCM02}, and military
planning~\cite{GhareMT71} (see also the references in~\cite{Zenklusen15}). 
Consequently, interdiction problems have been extensively studied, especially in the
Operations Research literature; besides MST-interdiction, some well-studied interdiction 
problems include 
network-flow interdiction~\cite{Phillips93,Wood93,BurchCKMPS03,Zenklusen10b,GuruganeshSS15,ChestnutZ15}, 
shortest $s$-$t$ path interdiction~\cite{FulkersonH77,IsraeliW02,KhachiyanBBEGRZ08,Lee16}, 
and maximum-matching interdiction~\cite{Zenklusen10a,DinitzG13}.   
All these problems, as well as MST-interdiction, are \nphard.

\paragraph{Our results.}
Our main result is a $4$-approximation algorithm for MST interdiction
(Theorem~\ref{mstinterthm}), i.e., we compute in polytime a solution of value at least
(optimum)/4.    
This constitutes a substantial improvement over the previous-best approximation ratio of    
$14$ obtained by Zenklusen~\cite{Zenklusen15}. 

Notably, and perhaps more importantly, our algorithm is simple, and its analysis   
is significantly simpler and cleaner than the one in~\cite{Zenklusen15}. 
The key ingredient (see also ``Our techniques'') of both our algorithm and
the one in~\cite{Zenklusen15} is a procedure for extracting a good interdiction set from
one that exceeds the interdiction budget. Whereas~\cite{Zenklusen15} uses a greedy
algorithm with a rather involved analysis to achieve this, 
our simple and more-effective procedure is based on two chief insights.
First, 
we discern that the key combinatorial aspects of this ``extraction problem'' can be
captured quite nicely via a clean generalization of the knapsack problem 
called the {\em tree knapsack problem}~\cite{JohnsonN83} (Section~\ref{tkn}). 
In particular, we argue that approximation guarantees for tree knapsack 
{\em relative to the natural LP for this problem} translate directly to guarantees for MST
interdiction. 
Second, complementing the above insight, we show that the tree knapsack problem admits a
simple {\em iterative-rounding} based algorithm that achieves a strong LP-relative
guarantee (Theorem~\ref{tknthm}, Corollary~\ref{strongtkn}). 
Our improved guarantee for MST interdiction then readily follows by combining these
two ideas.

We also show a lower bound of 3 (Theorem~\ref{theo:integrality_gap}) on the approximation
ratio achievable relative to the upper bound used in our analysis (and the analysis
in~\cite{Zenklusen15}), thereby showing that our algorithm and analysis are nearly tight.  

Our MST-interdiction result also yields an improved guarantee for the
{\em metric-TSP interdiction} problem (Section~\ref{tspinter}), wherein we have 
{\em metric} edge weights $\{w_e\}$, 
and we seek an
interdiction set $R$ with $\sum_{e\in R}c_e\leq B$ so as to maximize the minimum
$w$-weight of a closed walk in $G-R$ that visits all nodes at least once. 
Since an $\al$-approximation for MST interdiction yields a $2\al$-approximation for metric-TSP 
interdiction~\cite{Zenklusen15}, we obtain an approximation factor of $8$ for metric-TSP
interdiction, which improves upon the previous-best factor of $28$~\cite{Zenklusen15}. 

In Section~\ref{maxstinter}, we consider the {\em maximum-spanning-tree interdiction}
problem, where the goal is to {\em minimize the maximum $w$-weight of a spanning tree
of $G-R$}. We show that this problem is at least as hard to approximate as the minimization
version of the {\em densest-$k$-subgraph} problem (\mindks). \mindks does not admit any
constant-factor approximation under certain less-standard complexity
assumptions~\cite{RaghavendraS10} (and is believed to have a larger
inapproximability threshold), so this highlights a stark contrast with the
{MST-interdiction problem.}

\paragraph{Our techniques.}
We give an overview of our algorithm for MST interdiction.
Let $\val(R)$ be the $w$-weight of an MST of $G-R$. 
Using standard arguments, we can reduce the problem to the following setting (see
Section~\ref{prelim} and Theorem~\ref{bipoint}): we are given interdiction sets 
$R_1\sse R_2$ with $c(R_1)<B<c(R_2)$ such that $a\cdot\val(R_1)+b\cdot\val(R_2)\geq\OPT$,
where $\OPT$ is the
optimal value and $a,b\geq 0$ are such that $a+b=1$ and 
$a\cdot c(R_1)+b\cdot c(R_2)=B$. 
These arguments resemble the ones in~\cite{Zenklusen15}, but 
we do not need to assume that the $w_e$ weights are powers of $2$. (We
emphasize however that this by itself is not the chief source of our improvement.)
The technical meat of the algorithm, and where we diverge significantly
from~\cite{Zenklusen15} to obtain our improved guarantee, is to show how to extract a good 
interdiction set from $R_1, R_2$. As mentioned earlier, we replace the greedy algorithm
of~\cite{Zenklusen15} for extracting a good interdiction set from $R_2$, and its
associated intricate analysis, by considering the tree-knapsack problem to capture the key
aspects of this extraction problem, and devise a simple iterative-rounding algorithm that
yields a strong LP-relative guarantee for tree knapsack. This conveniently 
translates to a much-improved 5-approximation algorithm for MST interdiction
(Theorem~\ref{5apxthm}). The further improvement to a 4-approximation arises by also
leveraging $R_1$ to find a good interdiction set: instead of focusing solely on $R_2$ (as
done in~\cite{Zenklusen15}), we return an interdiction set $R$ such that 
{$R_1\sse R\sse R_2$ (see Section~\ref{4apx}).}

To arrive at the tree knapsack problem, 
observe that 
$\val(R)$ can be conveniently expressed as a weighted sum
of the number of components of $(V,\{e\in E\sm R: w_e\leq t\})$, where $t$ ranges over
some distinct edge weights, say, $0\leq w_1<\dots<w_k$ 
(Lemma~\ref{lem:val}). 
Let $\A_0$ denote the components of $(V,E_{\leq 0}:=\es)$, and
$\A_i$ denote the components of $(V,E_{\leq i}:=\{e\in E\sm R_2: w_e\leq w_i\})$ for
$i=1,\ldots,k$. 
The multiset $\bigcup_{i=0}^k \A_i$ forms a laminar family, 
which can be viewed as a rooted tree. We seek to build our interdiction set $R$ by
selecting a suitable collection of sets from this laminar family, ensuring that if we
pick a component $A\in\A_i$, then $\dt(A)\cap E_{\leq i}$ is included in $R$
(so that $A$ is indeed a component of $(V,E_{\leq i}\sm R)$). 
Whereas $\val(R)$ is nicely {\em decoupled} across the selected components, 
it is harder to decouple the interdiction cost incurred and account for it.
For instance, summing $c\bigl(\dt(A)\cap E_{\leq i}\bigr)$ for each selected $A\in\A_i$
may grossly overestimate the interdiction cost, whereas 
summing $c\bigl(\dt(A)\cap\{e: w_e=w_i\}\bigr)$ for each selected $A\in\A_i$ 
{\em underestimates} the interdiction cost.
A crucial insight is that, if we ensure that {\em whenever we pick $A\in\A_i$, we   
also pick its children in the laminar family}, 
then {\em summing $c\bigl(\dt(A)\cap\{e: w_e=w_i\}\bigr)$ for each selected
$A\in\A_i$ is a good proxy for the interdiction cost incurred.}

This motivates the definition of the tree knapsack problem: given a rooted tree $\Gm$ with
node values $\{\alpha_v\}$, node weights $\{\beta_v\}$, and budget $B$, we want to pick a
maximum-value  
{\em downwards-closed} set of nodes (not containing the root) whose weight is at most $B$, 
where downwards-closed means that if we pick a node, then we also pick all its
children. 
The standard knapsack problem is thus the special case where $\Gm$ is a
star (rooted at its center). We consider the natural LP \eqref{tknlp} for tree knapsack,
and generalizing a well-known result for knapsack, show that we can efficiently compute a
solution of 
value at least $\opttknlp-\max_{\text{chains $C$}}\sum_{v\in C}\al_v$
(Theorem~\ref{tknthm}),  where a chain is a subset of a root-leaf path.

Finally, we show that for the tree-knapsack instance derived (as above) from
$R_2$, $\opttknlp$ is ``large'' (Lemma~\ref{lem:frac}), and combining this with the above
bound yields our approximation guarantee.

\paragraph{Related work.} 
MST interdiction in its full generality seems to have been first considered
by~\cite{LinC93}, who showed that the problem is \nphard. The approximation question for
MST interdiction was first investigated by~\cite{FredericksonS99}. 
They focused on the setting with unit interdiction costs, often called the 
{\em $B$-most-vital-edges} problem, showed that this special case
remains \nphard, 
and obtained an $O(\log B)$-approximation (which also yields an 
$O(\log |E|)$-approximation with general interdiction costs).  
This 
guarantee was improved only recently by Zenklusen~\cite{Zenklusen15},
who gave the first (and previous-best) $O(1)$-approximation algorithm for (general) MST
interdiction, 
achieving an approximation ratio of $14$. The $B$-most-vital edges problem has been well
studied for $B=1$ and for $B=O(1)$, where it can be solved optimally; 
see, e.g.,~\cite{Liang01} and the references therein. The special case of MST interdiction
where we have only two distinct edge weights captures the {\em budgeted graph disconnection}
(\bgd) problem~\cite{EngelbergKLN07} for which a $2$-approximation is
known~\cite{EngelbergKLN07}. As noted by~\cite{Zenklusen15}, MST interdiction can be
viewed as multilevel-\bgd, which makes it much more challenging 
as it is
difficult to control the interactions at the different levels. It is noteworthy that our
approximation ratio of $4$ for MST interdiction is quite close to the approximation
ratio of $2$ for \bgd.  

As with MST interdiction, until recently, 
there were wide gaps in our understanding of the approximability 
of the other classic \nphard interdiction problems mentioned earlier. 
Maximum $s$-$t$ flow interdiction, even on undirected graphs with unit
interdiction costs, is now known to be at least as hard as \mindks on $\ld$-uniform
hypergraphs. This follows from a recent hardness result for {\em $k$-route $s$-$t$ cut} 
in~\cite{GuruganeshSS15}, which turns out to be an equivalent problem.% 
\footnote{In $k$-route $s$-$t$ cut, the goal is to remove a min $w$-cost set of edges so as
to reduce the $s$-$t$ edge connectivity to at most $k-1$. 
This corresponds to taking all but the $k-1$ most-expensive edges of some cut. 
So we can rephrase this problem as follows: remove at most $k-1$
edges to minimize the (min-$s$-$t$-cut value = max-$s$-$t$-flow value) with capacities
$\{w_e\}$; 
this is precisely the maximum $s$-$t$ flow interdiction problem with 
{unit interdiction costs and budget $k-1$.}}     
This hardness result has been rediscovered (in a slightly weaker
form) by~\cite{ChestnutZ15}, who also gave an $O(n)$-approximation algorithm.
For shortest $s$-$t$ path interdiction, very recently, Lee~\cite{Lee16} proved a
super-constant hardness result. For maximum-matching interdiction, \cite{DinitzG13}
devised the first $O(1)$-approximation algorithm. 
Despite this recent progress, interdiction variants of common optimization problems are
generally not well understood, especially from the viewpoint of approximability.

The tree knapsack problem was introduced by~\cite{JohnsonN83}, and is a special case of
the {\em partially-ordered knapsack} (\pok) problem~\cite{KolliopoulosS07}. While an
FPTAS can be obtained for tree knapsack and some special cases of \pok via dynamic
programming~\cite{JohnsonN83,KolliopoulosS07}, and the natural LP for 
\pok has been investigated~\cite{KolliopoulosS07}, our LP-relative guarantee and rounding
{algorithm for tree knapsack are new.}

\section{Preliminaries} \label{prelim} 

For any vector $d\in\R^E$ and any subset $F\sse E$ of edges, we use $d(F)$ to denote
$\sum_{e\in F}d_e$. 
Given a subset $R\subseteq E$ of edges, we use $\val(R)$, which we call the {\em value} of
$R$, to denote the $w$-weight of an MST in the multigraph $G-R$, 
i.e., $\val(R):=\min_{\text{spanning trees $T$ of $G-R$}}w(T)$. 
The {\em minimum-spanning-tree interdiction} problem can thus be restated as 
$\max\ \bigl\{\val(R): R \subseteq E,\ c(R) \le B \bigr\}$.

If there is an interdiction set $R$ with $c(R)\leq B$ such that $G-R$ is disconnected,
then $\val(R) = \infty$, and so the MST-interdiction problem is unbounded. 
Note that this happens iff a min-cut $\dt(S)$ of $G$ satisfies 
$c\bigl(\dt(S)\bigr)\leq B$, and we can efficiently detect this.
So in the sequel, we assume that this is not the case.
Let $\OPT$ denote the optimal value of the MST-interdiction problem (which is now finite). 
For $F\sse E$, let $\sg(F)$ denote the number of connected components of $(V,F)$.

Let $w_1,w_2,\ldots,w_M$ be the distinct weights in $\{w_e: e\in E\}$, where 
$0\leq w_1<w_2<\dots<w_M$. For $i=1,\ldots,M$, define 
$E_i:=\{e\in E: w_e=w_i\}$ and $E_{\leq i}:=\{e\in E: w_e\leq w_i\}$. For notational convenience, 
we define $w_0:=0$ and $E_0=E_{\leq 0}:=\es$. 
(Note that $E_0$ is {\em not} necessarily $\{e\in E: w_e=w_0\}$, and $E_{\leq 0}$ is not
necessarily $\{e\in E: w_e\leq w_0\}$.) 

Let $k\in\{1,\ldots,M\}$ be the {\em smallest} index such that 
$c\bigl(\dt(S)\cap E_{\leq k}\bigr) >B$ for every $\es\neq S\subsetneq V$; that is, the
multigraph $(V,E_{\leq k}\sm R)$ is connected for all $R$ such that $c(R)\leq B$. 
Note that $k$ is well defined due to our earlier assumption. This implies the following
properties, as also observed in~\cite{Zenklusen15}:
(i) $\OPT\geq w_k$ (since, by definition of $k$, there is a feasible interdiction set $R$ whose removal
disconnects $(V,E_{\leq k-1})$);   
(ii) for any $R$ with $c(R)\leq B$, we have $\val(R)=\val(R\cap E_{\leq k-1})$, and hence,
there is an optimal solution that only interdicts edges from $E_{\leq k-1}$; and 
(iii) given (ii), we may add additional edges of weight $w_k$ without impacting the
optimal value, so we may assume that $(V,E_k)$ 
is connected.
We summarize these properties and assumptions below. 

\begin{claim} \label{ass:connected}
Let $k\in\{1,\ldots,M\}$ be the smallest index such that 
$(V,E_{\leq k}\sm R)$ is connected for every $R \subseteq E$ with $c(R) \le B$. Assume
that such a $k$ exists. Then, (i) $\OPT\geq w_k$, and (ii) there is an optimal solution
$R^*$ such that $R^*\sse E_{\leq k-1}$. Moreover, we may assume that (iii) the multigraph 
$(V,E_k)$ is connected.
\end{claim}

\begin{lemma} \label{lem:val}
Let $R\subseteq E$ be an edge-set such that $(V,E_{\leq k}\sm R)$ is connected. Then
$\val(R) = - {w}_k + \sum_{i =0}^{k-1} \sigma\left(E_{\le i} \setminus R\right)({w}_{i+1} - {w}_i)$.
\end{lemma}

\begin{proof}
Consider, for example, running Kruskal's algorithm to obtain an MST of $G-R$.
We include exactly $\sg(E_{\leq j-1} \setminus R)-\sg(E_{\leq j} \setminus R)$ edges of
weight $w_j$ for every $1 \le j \le M$, and this 
quantity is 0 for all $j>k$. It follows that  
\begin{align*}
\val(R) & = \sum_{j=1}^{M}\Bigl(\sigma\left(E_{\le j-1} \setminus R\right) - \sigma(E_{\le j} \setminus R)\Bigr){w}_j 
= \sum_{j=1}^k\Bigl(\sg(E_{\leq j-1}\sm R)-\sg(E_{\leq j}\sm R)\Bigr)w_j \\
& = \sum_{j=1}^{k}\Bigl(\sigma\left(E_{\le j-1} \setminus R\right) - \sigma(E_{\le j}\setminus R)\Bigr)
\sum_{i=0}^{j-1} ({w}_{i+1} - {w}_{i}) \\
& = \sum_{i=0}^{k-1}({w}_{i+1}-{w}_i)
\sum_{j=i+1}^{k}\Bigl(\sigma (E_{\le j-1} \setminus R) - \sigma(E_{\le j} \setminus R)\Bigr) \\
& = \sum_{i=0}^{k-1}\bigl(\sigma(E_{\le i} \setminus R) - 1\bigr)({w}_{i+1} - {w}_i)
= -w_k + \sum_{i=0}^{k-1}\sigma(E_{\le i} \setminus R)({w}_{i+1} - {w}_i). \ \ \qedhere
\end{align*}
\end{proof}

Given Claim~\ref{ass:connected}, we focus on interdiction sets $R\sse E_{\leq k-1}$
and recast the MST-interdiction problem as: 
$\max\ \bigl\{\val(R): R\sse E_{\leq k-1},\ c(R)\leq B\bigr\}$. As is common in the study
of constrained optimization problems (see, e.g.,~\cite{KonemannPS11,GrandoniRSZ14}
and the references therein), we Lagrangify the budget constraint $c(R)\leq B$,
and consider the following Lagrangian problem (offset by $-\ld B$), 
{where $\ld\geq 0$ is a parameter:} 
\begin{equation}
\max_{R\sse E_{\leq k-1}} \quad f_\ld(R):=\val(R)-\ld c(R). \tag{P$_{\ld}$} 
\label{lagp} 
\end{equation}
The expression for $\val(R)$ in Lemma~\ref{lem:val} holds for all $R\sse E_{\leq k-1}$ as
$(V,E_k)$ 
is connected. Since 
$\sg(E_{\leq i}\sm R)$ is a supermodular function of $R$, this implies that $\val(\cdot)$,
and hence the objective function $f_\ld(\cdot)$ of \eqref{lagp}, 
is {\em supermodular} over the domain $2^{E_{\leq  k-1}}$: for any   $A_1,A_2\sse E_{\leq k-1}$, we have
$f_\ld(A_1)+f_\ld(A_2)\leq f_\ld(A_1\cap A_2)+f_\ld(A_1\cup A_2)$. 
Hence, \eqref{lagp} can be solved exactly, which we crucially exploit. 

Let $\Oc^*_\ld$ denote the set of optimal solutions to \eqref{lagp}.
Observe that for any $\ld\geq 0$ and any $R\in\Oc^*_\ld$, 
we have $\val(R)-\ld c(R)\geq\OPT-\ld B$.
So if we find some $\ld\geq 0$ and $R\in\Oc^*_\ld$ such that 
$c(R)=B$, we have $\val(R)\geq\OPT$, so $R$ is an optimal solution. 
In general, such a pair $(\ld,R)$ need not exist, or can be hard to find. However, by doing a
binary search for $\ld$, or alternatively, as noted in~\cite{Zenklusen15}, via parametric
submodular-function minimization~\cite{FleischerI03,Nagano07}, we can obtain 
the following result; we include a self-contained proof in Appendix~\ref{append-prelim}.

\begin{theorem}[\cite{Zenklusen15}] \label{bipoint}
One can find in polytime: either (i) an optimal solution to the MST-interdiction problem,  
or (ii) a parameter $\ld \ge 0$ and two optimal solutions $R_1, R_2$ to \eqref{lagp} such
that $R_1 \subseteq R_2$ and  $c(R_1)<B<c(R_2)$. 
\end{theorem}

\section{The tree knapsack problem} \label{tkn} 
We now define the {\em tree knapsack problem}, and devise a simple, clean \mbox{LP-based}
approximation algorithm for this problem (Theorem~\ref{tknthm}, Corollary~\ref{strongtkn}). 
As we show in Section~\ref{mstinter}, the tree knapsack problem nicely abstracts the key
combinatorial problem encountered in extracting a good interdiction set from an
{over-budget} set $R_2$ in case (ii) of Theorem~\ref{bipoint}, 
and our LP-relative guarantees for tree knapsack readily yield improved approximation
{guarantees for MST interdiction.}

In the tree knapsack problem~\cite{JohnsonN83}, we have a tree $\Gm=(\{r\}\cup N,A)$
rooted at node $r$. 
Each node $v\in N$ has a {\em value} $\al_v\geq 0$ and a {\em weight}
$\beta_v\geq 0$, and we have a budget $B$.  
We say that a subset $S\subseteq N$ of nodes is {\em downwards-closed} if for every 
$v\in S$, all children of $v$ are also in $S$. The goal is to find a maximum-value
downwards-closed set $S\sse N$ (so $r\notin S$) such that $\sum_{v\in S}\beta_v\leq B$.
Observe that the (standard) knapsack problem is precisely the special case of tree
knapsack where the underlying tree is a star (rooted at its center). Throughout, we use
$v$ to index nodes in $N$. For $S\sse N$ and a vector $\rho\in\R^N$, we use $\rho(S)$ to
denote  $\sum_{v\in S}\rho_v$.  

The following 
is a natural LP-relaxation for the tree knapsack problem involving variables $x_v$ for all
$v$. Let $\child(v)$ denote the set of children of node $v$.  
\begin{alignat}{3}
\max & \quad & \sum_v \al_vx_v & \tag{TK-P} \label{tknlp} \\
\text{s.t.} && x_v & \leq x_u \qquad && \frall v, \frall u\in\child(v) 
\label{const:down} \\
&& \sum_v \beta_vx_v & \leq B, 
&& 0 \leq x_v \leq 1 \quad \frall v. \notag 
\end{alignat}
Tree knapsack was first defined by~\cite{JohnsonN83} who devised an FPTAS for this problem 
via dynamic programming. However, for our purposes, we need an approximation
guarantee relative to the above LP, which was not known previously.

The main result of this section is as follows. We say that $C\sse N$ is a {\em chain} if for
every two distinct nodes in $C$, one is a descendant of the other. 

\begin{theorem} \label{tknthm}
We can compute in polytime an integer solution to \eqref{tknlp} of value at least
$\OPT_{\text{\ref{tknlp}}}-\max_{\text{chains $C\sse N$}}\alpha(C)$.
\end{theorem}

Theorem~\ref{tknthm} nicely generalizes a well-known result about the standard knapsack
problem, namely, that we can always obtain a solution of value at least
(LP-optimum)$-\max_v\al_v$. Notice that when $\Gm$ is a star (i.e., we have a knapsack
instance), this is {\em precisely} the guarantee that we obtain from the theorem.
The proof of Theorem~\ref{tknthm} relies on the following structural result (which
extends a similar result known for knapsack). 
Let $\Gm(v)$ denote the \nolinebreak \mbox{subtree of $\Gm$ rooted at $v$}. 

\begin{lemma} \label{lem:extreme_point}
Let $\bar{x}$ be an extreme-point solution to the linear program \eqref{tknlp}. Then there is
at most one child $v$ of $r$ for which the subtree $\Gm(v)$ contains a fractional node,
i.e., some node $w$ with $0<\bx_w<1$.
\end{lemma}

\begin{proof} 
Suppose for a contradiction that the root $r$ has two children $v_1$ and $v_2$ such that
the subtrees $\Gm(v_1)$ and $\Gm(v_2)$ both contain at least one fractional node. We show
that for some nonzero vector $d \in \mathbb{R}^{N}$, the solutions $\bx\pm d$ are feasible
to \eqref{tknlp}, which contradicts that $\bx$ is an extreme point.

For $j=1, 2$, let $N_j$ be a maximal set of nodes in the subtree $\Gm(v_j)$ such that:
(a) $N_j$ induces a connected subgraph of $\Gm(v_j)$; and 
(b) all nodes in $N_j$ have the same $x_w$ value, which is fractional.
We will always set $d_v=0$ for all $v\notin N_1\cup N_2$.
Note that for any $\mu_1,\mu_2\in\R$ with sufficiently small absolute value, if we set
$d_v=\mu_1$ for all $v\in N_1$ and $d_v=\mu_2$ for all $v\in N_2$, then the vectors 
$x\pm d$ satisfy constraints \eqref{const:down} (due to the maximality of $N_1, N_2$),
and $0\leq (x\pm d)_v\leq 1$ for all $v\in N$. 

We argue that we can choose suitably small $\mu_1, \mu_2$ (not both equal to zero) so that  
$\sum_v\beta_vd_v=0$, and so $x\pm d$ also satisfy the budget constraint, 
and hence are feasible to \eqref{tknlp}. 
If $\beta(N_1)=0$, if we take a sufficiently small $\mu_1>0$ and $\mu_2=0$, then clearly
$\sum_v\beta_vd_v=0$.
Otherwise, for $\e>0$ and suitably small, we take $\mu_1=\e\beta(N_2)$ and
$\mu_2=-\e\beta(N_1)$. Then again,
$\sum_v\beta_vd_v=\e\beta(N_1)\beta(N_2)-\e\beta(N_2)\beta(N_1)=0$  
(so $x\pm d$ is feasible to \eqref{tknlp}).
\end{proof}

\begin{proofof}{Theorem~\ref{tknthm}}
We use {\em iterative rounding}, and the proof is by induction on the depth $d$ of $\Gm$,
which is the maximum number of edges on a root-leaf path. 

If $d=0$, then $N=\es$, and \eqref{tknlp} has no variables and constraints, so the
statement is vacuously true.  
So suppose $d\geq 1$.
Let $x^*$ be an extreme-point optimal solution of \eqref{tknlp}. 
If $x^*$ is integral,
then we obtain value $\opttknlp$, completing the induction step. Otherwise, by
Lemma~\ref{lem:extreme_point}, there is exactly one child $v$ of $r$ such that the subtree $\Gm(v)$  
contains a fractional node. 

Set $\tx'=\restr{x^*}{N\sm\Gm(v)}$, i.e., $x^*$ restricted to $N\sm\Gm(v)$, which is
integral. 
We have $\sum_{u\in N\sm\Gm(v)}\al_u\tx'_u=\opttknlp-\sum_{w\in\Gm(v)}\al_wx^*_w$. 
Now consider the tree knapsack instance defined by the tree $\Gm(v)$ with root $v$, and
budget $B-\sum_{u\in N\sm\Gm(v)}\beta_u\tx'_u$ (and values $\al_w$ and weights $\beta_w$
for all $w\in\Gm(v)\sm\{v\}$). 
Observe that $\restr{x^*}{\Gm(v)\sm\{v\}}$ 
is a fractional solution to the LP-relaxation \eqref{tknlp} corresponding to this tree
knapsack problem, so the optimal value of this LP is at least 
$\sum_{w\in\Gm(v)\sm\{v\}}\al_wx^*_w$. (These objects are null if $\Gm(v)=\{v\}$.)
Thus, since $\Gm(v)$ has depth at most $d-1$, by our induction hypothesis, 
our rounding procedure applied to this tree knapsack instance
yields an integer solution $\tx''\in\{0,1\}^{\Gm(v)\sm\{v\}}$ of value at least 
$\sum_{w\in\Gm(v)\sm\{v\}}\al_wx^*_w-\max_{\text{chains $C\sse\Gm(v)\sm\{v\}$}}\al(C)$.
Thus, taking $\tx=(\tx',\tx_v=0,\tx'')$, we obtain a feasible integer solution to
\eqref{tknlp} having value at least 
\begin{equation*}
\begin{split}
\opttknlp & -\sum_{w\in\Gm(v)}\al_wx^*_w+\sum_{w\in\Gm(v)\sm\{v\}}\al_wx^*_w
-\max_{\text{chains $C\sse\Gm(v)\sm\{v\}$}}\al(C) \\
& \geq \opttknlp-\al_v-\max_{\text{chains $C\sse\Gm(v)\sm\{v\}$}}\al(C) \\
& \geq \opttknlp-\max_{\text{chains $C\sse N$}}\al(C) \ .
\end{split}
\end{equation*}
This completes the induction step, and hence the proof of the theorem.
\end{proofof}

We remark that (as is standard) the iterative-rounding procedure in Theorem~\ref{tknthm}
is in fact combinatorial, since when we move to the subtree $\Gm(v)$, we only need to
move from $\restr{x^*}{\Gm(v)\sm\{v\}}$ to an extreme-point of the LP 
of the smaller tree-knapsack instance of no smaller value (instead of obtaining an optimal LP
solution), which can be done combinatorially 
{(as in the proof of Lemma~\ref{lem:extreme_point}).}

We now prove a somewhat stronger version of Theorem~\ref{tknthm} that will be useful in
Section~\ref{mstinter}, where we utilize tree knapsack to solve the
\mbox{MST-interdiction} problem. 
The depth of a node $v$ is the number of edges on the (unique) $r$-$v$ path of $\Gm$. 
Let $\level_i(\Gm)$ be the set of nodes of $\Gm$ at depth $i$; we drop $\Gm$ if it is
clear from the context. For a chain $C$ of $\Gm$, let $C_i$ denote $C\cap\level_i(\Gm)$;
note that $|C_i|\leq 1$.

\begin{corollary} \label{strongtkn}
We can obtain in polytime an integer solution $\tx$ to \eqref{tknlp} of value at least
$\OPT_{\text{\ref{tknlp}}}
-\max_{\text{chains $C\sse N$}}\bigl\{\sum_{i\geq 1:\tx(\level_i)<|\level_i|}\alpha(C_i)\bigr\}$.
\end{corollary} 

\begin{proof} 
The result follows from the proof of Theorem~\ref{tknthm} via a more-careful accounting.
Recall that we use induction on the depth $d$ of $\Gm$. The base case when $d=0$ is again
vacuously true. So suppose $d\geq 1$, and let $x^*$ be an extreme-point optimal solution to
\eqref{tknlp}. If $x^*$ is integral, we are done, so suppose that there is a child $v$ of
$r$ such that $\Gm(v)$ contains a fractional node.
In the sequel, $\level_i$ denotes $\level_i(\Gm)$.

As before, let $\tx'=\restr{x^*}{N\sm\Gm(v)}$, $\tx_v=0$, and let $\tx''$ be the integer
solution obtained by induction for the tree knapsack instance defined by the tree $\Gm(v)$
with root $v$ and budget $B-\sum_{u\in N\sm\Gm(v)}\beta_u\tx'_u$. Therefore, $\tx''$ has
value at least  
\begin{equation}
\sum_{w\in\Gm(v)\sm\{v\}}\al_wx^*_w
-\max_{\text{chains $C\sse\Gm(v)\sm\{v\}$}}
\sum_{\substack{i\geq 1: \\ \tx''(\level_i(\Gm_v))<|\level_i(\Gm_v)|}}\al(C_i).
\label{corineq2}
\end{equation}
Thus, $\tx=(\tx',\tx_v,\tx'')$ is a feasible integer solution to \eqref{tknlp} of value at
least 
\begin{equation*}
\begin{split}
\opttknlp & -\sum_{w\in\Gm(v)}\al_wx^*_w+\sum_{w\in\Gm(v)\sm\{v\}}\al_wx^*_w
-\max_{\text{chains $C\sse\Gm(v)\sm\{v\}$}}
\sum_{\substack{i\geq 1: \\ \tx''(\level_i(\Gm_v))<|\level_i(\Gm_v)|}}\al(C_i) \\
& \geq \opttknlp-\al_v
-\max_{\text{chains $C\sse\Gm(v)\sm\{v\}$}}
\sum_{\substack{i\geq 1: \\ \tx''(\level_i(\Gm_v))<|\level_i(\Gm_v)|}}\al(C_i) \\
& \geq \opttknlp-\max_{\text{chains $C\sse N$}}
\sum_{\substack{i\geq 1: \\ \tx(\level_i)<|\level_i|}}\alpha(C_i).
\end{split}
\end{equation*}
The last inequality above follows by noting that for any chain $C\sse \Gm(v) \sm \{v\}$, letting $C' := C \cup \{v\}$ (which is also a chain), we have
$$
\sum_{\substack{i\geq 1: \\ \tx(\level_i)<|\level_i|}}\alpha(C'_i)=
\al_v+\sum_{\substack{i\geq 1: \\ \tx''(\level_i(\Gm_v))<|\level_i(\Gm_v)|}}\al(C_i). 
$$ 
This completes the induction step, and hence the proof.
\end{proof}

\section{MST interdiction} \label{mstinter}
Our main technical result is the following theorem.

\begin{theorem} \label{mstinterthm} \label{theo:MST_interdiction}
There is a $4$-approximation algorithm for MST interdiction. 
\end{theorem}

The above guarantee substantially improves the previous-best approximation ratio of $14$
obtained by~\cite{Zenklusen15}. 
Also, notably and significantly, our algorithm and analysis, which are
based on the tree knapsack problem introduced in Section~\ref{tkn}, are noticeably
simpler and cleaner than the one in~\cite{Zenklusen15}. 
Improved guarantees for MST interdiction readily follow from (Theorem~\ref{tknthm} and)
Corollary~\ref{strongtkn} and Lemma~\ref{newtkn}, yielding approximation ratios of $5$ and
$4$ respectively for MST interdiction (see Theorem~\ref{5apxthm} and Section~\ref{4apx}).  
The proof below shows a slightly worse guarantee of $5$ but introduces the main underlying 
ideas. Section~\ref{4apx} discusses the refinement needed to obtain the $4$-approximation.  
 
Our algorithm follows the same high-level outline as the one in~\cite{Zenklusen15}. 
As mentioned earlier, we consider the Lagrangian problem \eqref{lagp}, 
$\max_{R\sse E_{\leq k-1}} f_\ld(R):=\val(R)-\ld c(R)$, obtained by dualizing 
the budget constraint $c(R)\leq B$. We then utilize Theorem~\ref{bipoint}. If this returns 
an optimal solution, then we are done. So assume in the sequel that Theorem~\ref{bipoint}
returns $\ld \geq 0$ and two optimal solutions $R_1$ and $R_2$ to \eqref{lagp} such that
$R_1\sse R_2$ and $c(R_1)<B<c(R_2)$. 

For $R\sse E_{\leq k-1}$, define 
$h(R):=\sum_{i =0}^{k-1} \sigma\left(E_{\le i}\setminus R\right)({w}_{i+1}-{w}_i)=\val(R)+w_k$. 
Let $R^*\sse E_{\leq k-1}$ denote an optimal
solution to the MST-interdiction problem, so $\OPT=h(R^*)-w_k$. Let $a,b\geq 0$ such that
$a+b=1$ and $ac(R_1)+bc(R_2)=B$. Then, since 
$\val(R_1)-\ld c(R_1)=\val(R_2)-\ld c(R_2)\geq\OPT-\ld B$, we have 
$ah(R_1)+bh(R_2)\geq h(R^*)$. We establish our approximation guarantee by comparing the
value of our solution against the upper bound $ah(R_1)+bh(R_2)-w_k$.  
The following claim shows that this upper bound is precisely the optimal value of the
Lagrangian relaxation of the MST interdiction problem, which is 
$\UB:=\min_{\ld'\geq 0}\bigl(\ld' B+\max_{R\sse E_{\leq k-1}}f_{\ld'}(R)\bigr)$.
Complementing our $4$-approximation, in Section~\ref{lbound}, we prove a lower bound of
$3$ on the approximation ratio achievable relative to $\UB$.

\begin{claim}
\label{claim:def_ub}
We have $ah(R_1) + bh(R_2) - w_k = \UB$.
\end{claim} 

\begin{proof} 
Let $\OPT(\text{P}_{\ld'})$ denote the optimal value of the subproblem
\begin{equation}
\max_{R\sse E_{\leq k-1}} \quad f_{\ld'}(R):=\val(R)-\ld' c(R). \tag{P$_{\ld'}$} 
\end{equation}
Define $\eta(\ld') :=  \OPT(\text{P}_{\ld'}) + \ld' B$. So we have $\UB = \min_{\ld' \ge 0} \eta(\ld')$.
We have
\begin{align*}
ah(R_1) + bh(R_2) - w_k & = a\cdot\val(R_1)+b\cdot\val(R_2)
= a(\val(R_1) - \ld c(R_1)) + b(\val(R_2) - \ld c(R_2)) + \ld B \\
& = (a + b) OPT(\text{P}_{\ld}) + \ld B = \eta(\ld) \ge \UB.
\end{align*}

We now argue that $\UB\geq ah(R_1)+bh(R_2)-w_k$ by showing that $\eta(\ld') \ge \eta(\ld)$
for every $\ld' \ge 0$. We have
\begin{align*}
\eta(\ld') & = \OPT(\text{P}_{\ld'}) + \ld'B \ge 
\max\bigl\{\val(R_1) - \ld' c(R_1),\val(R_2)-\ld' c(R_2)\bigr\}+ \ld'B \\ 
& \geq a(\val(R_1)-\ld' c(R_1))+b(\val(R_2)-\ld' c(R_2))+\ld' B \\
& = a\cdot\val(R_1)+b\cdot\val(R_2) = \eta(\ld). \qedhere
\end{align*}
\end{proof}

\paragraph{Translation to tree knapsack.}
We now describe how the problem of combining $R_1$ and $R_2$ to extract a good, feasible
interdiction set can be captured by a suitable instance of the tree knapsack problem
defined in Section~\ref{tkn}. 

For $i=0,\ldots,k$, let $\mathcal{A}_i \subseteq 2^V$ be the partition of $V$ induced by
the connected components of the multigraph $(V, E_{\le i} \setminus R_2)$. 
Thus, $\A_k=\{V\}$ and $\A_0=\{\{v\}: v\in V\}$. 
The multiset $\bigcup_{i=0}^k\A_i$, where we include $S\sse V$ multiple times if
it lies in multiple $\A_i$s, is a {\em laminar family} (i.e., any two sets in the
collection are either disjoint or one is contained in the other). 
This laminar family can naturally be viewed as a rooted tree, which defines the tree
$\Gm$ in the tree knapsack problem.  
Taking a cue from Lemma~\ref{lem:val}, we build our interdiction
set $R$ by selecting a suitable collection of sets from this laminar family, ensuring that
if we pick some $A\in\A_i$, then we include all edges of $\dt(A)\cap E_{\leq i}$ in $R$
and create $A$ as a component of $(V,E_{\leq i}\sm R)$ (and hence
contribute $w_{i+1}-w_i$ to $h(R)$). 
Formally, the tree $\Gm$ has a node $v^{A,i}$ for every component $A\in\A_i$ and
all $i=0,\ldots,k$. For $i>0$, the children of $v^{A,i}$ are the nodes 
$\{v^{S,i-1}: S\in\A_{i-1}, S\sse A\}$. Thus, $\Gm$ has depth $k$ and root
$r=v^{V,k}$. Recall that $L_i:=L_i(\Gm)$ denotes the set of nodes of $\Gm$ at depth $i$,
which correspond to the components in $\A_{k-i}$ here. Let $N$ be the set of non-root
nodes of $\Gm$. 

For a node $v^{A,i}\in N$ (so $0\leq i<k$), define its value 
$\alpha_{v^{A, i}}:=w_{i+1} - w_i$.   
Let $R(v^{A,i}):=\delta(A)\cap E_i$ (which is $\es$ for every leaf $v^{A,0}$). 
Define the weight of $v^{A,i}$ to be $\beta_{v^{A,i}}:=c\bigl(R(v^{A,i})\bigr)$. 
For $N'\sse N$, let $R(N'):=\bigcup_{q\in N'}R(q)$. Observe that $R(N) \subseteq R_2$.
We set the budget of the tree-knapsack instance to $B$, the 
{budget for MST interdiction.}  

The intuition is that we want to encode that picking node $v^{A,i}$ corresponds to
creating component $A$ in the multigraph $(V,E_{\leq i}\sm R)$, where $R$ is our
interdiction set, in which case $\al_{v^{A,i}}$ gives the contribution from $A$ to
$h(R)$. However, in order to pay for the interdiction cost $c\bigl(\dt(A)\bigr)$ incurred,
we need to take the $\beta_q$ weights of all nodes $q$ in the subtree rooted at $v^{A,i}$.
Therefore, we insist that if we pick $v^{A,i}$ then we pick all its descendants (i.e., 
we pick a downwards-closed set of nodes), and then 
$\sum_{q\in\Gm(v^{A,i})}\al_q$ 
gives the contribution from the components created to $h(R)$. 
Lemma~\ref{lem:translate} formalizes this intuition,
and shows that if $N'\sse N$ is a downwards-closed set of nodes, 
then $\beta(N')$ and $\al(N')$ are good proxies (roughly speaking) for 
the interdiction cost $c\bigl(R(N')\bigr)$ incurred and $h\bigl(R(N')\bigr)$
respectively. 

\begin{lemma} \label{lem:translate_cost} \label{lem:translate} \label{lem:translate_value1} 
Let $N' \subseteq N$ be downwards closed, and $R=R(N')$. Then 
\begin{enumerate}[(i), topsep=0.25ex, itemsep=0.25ex]
\item $\beta(N')/2\le c(R)\le\beta(N')$; and
\item $h(R)=\val(R)+w_k\geq\alpha(N') + \sum_{0 \le i \le k - 1: L_{k-i}\sm N'\neq\es} (w_{i+i} - w_i)$.
\end{enumerate}
\end{lemma}

\begin{proof}
Each edge in $R$ appears in at least one, and at most two, of the sets 
$\{R(q)\}_{q \in N'}$, so we obtain 
$\frac{1}{2} \sum_{q\in N'} c\bigl(R(q)\bigr)\le c(R)\le\sum_{q \in N'}c\bigl(R(q)\bigr)$. 
Part (i) follows by noting that $\beta(N')=\sum_{q \in N'} c\bigl(R(q)\bigr)$. 

For part (ii), consider an index $0 \le i \le k - 1$.
Since $N'$ is downwards closed, for every node $v^{A, i} \in N'$, all descendants of
$v^{A,i}$ are in $N'$; so $R\supseteq\dt(A) \cap E_{\le i}$ and $A$ is a connected component of
the multigraph $(V, E_{\le i} \setminus R)$. Further, note that if 
$L_{k-i}\sm N'\neq\es$, then the sets $\{A : v^{A,i} \in N'\}$ do not
cover $V$ entirely, and so $(V, E_{\le i} \setminus R)$ must have at least
one additional connected component. It follows that $(V,E_{\leq i}\sm R)$ always has at
least $\min\bigl\{|N' \cap\level_{k-i}| + 1, |\level_{k-i}|\bigr\}$ connected components.  
Plugging this in Lemma~\ref{lem:val} yields the result.
\end{proof}

\begin{lemma} \label{lem:fractional_solution} \label{lem:frac}
The vector $\hx:=\bigl(\hx_q=\frac{b}{2}\bigr)_{q\in N}$ is a feasible solution to
\eqref{tknlp} for the above tree-knapsack instance $(\Gm,\{\alpha_q\},\{\beta_q\},B)$. 
Hence, 
$\opttknlp\geq\frac{b}{2}\cdot h(R_2)$. 
\end{lemma}

\begin{proof}
It is clear that $\hx$ satisfies \eqref{const:down}, and $0\leq\hx_q\leq 1$ for all $q\in
N$. 
Applying Lemma~\ref{lem:translate_cost} to $N'=N$ (which is indeed downwards-closed),
we obtain $\beta(N)\le 2c\bigl(R(N)\bigr) \le 2c(R_2)$.
So $\sum_{q\in N}\beta_q\hx_q \le b\cdot c(R_2)\leq a\cdot c(R_1)+b\cdot c(R_2)=B$.
Finally, $\opttknlp$ is at least the objective value of $\hx$, which is
$\frac{b}{2}\cdot\al(N)=\frac{b}{2}\cdot h(R_2)$.
\end{proof}

Given this translation between the tree-knapsack and MST-interdiction problems, it is easy
to see that Corollary~\ref{strongtkn} (coupled with Lemmas~\ref{lem:translate}
and~\ref{lem:frac}) yields the following guarantee, which directly leads to an improved
approximation guarantee of $5$ for MST interdiction (see Claim~\ref{calcclm}).

\begin{lemma} \label{tkncor} 
We can obtain a feasible interdiction set $R$ such that 
$h(R)\geq\frac{b}{2}\cdot h(R_2)$. 
\end{lemma}

\begin{proof}
This is consequence of Corollary~\ref{strongtkn}, Lemma~\ref{lem:translate},
and Lemma~\ref{lem:fractional_solution}. 
Let $N'\sse N$ be the downwards-closed set
corresponding to the integer solution returned by Corollary~\ref{strongtkn}. Let
$R=R(N')$. We have 
$\al(N')\geq\frac{b}{2}\cdot h(R_2)
-\max_{\text{chains $C\sse N$}}\bigl\{\sum_{i\geq 1:\tx(\level_i)<|\level_i|}\alpha(C_i)\bigr\}$ 
by Corollary~\ref{strongtkn} and Lemma~\ref{lem:fractional_solution}.
To complete the proof, we apply part (ii) of Lemma~\ref{lem:translate}, noting that 
\begin{equation*}
\sum_{0 \le i \le k - 1: L_{k-i}\sm N'\neq\es} (w_{i+i} - w_i)
=\max_{\text{chains $C\sse N$}}\bigl\{\sum_{i\geq 1:\tx(\level_i)<|\level_i|}\alpha(C_i)\bigr\}
\qedhere
\end{equation*}
\end{proof}

\begin{claim} \label{calcclm}
We have $\max\bigl\{w_k,h(R_1)-w_k,\frac{b}{2}\cdot h(R_2)-w_k\bigr\}\geq\UB/5\geq\OPT/5$. 
\end{claim}

\begin{proof} 
We have
\begin{align*}
\max\Bigl\{w_k,h(R_1)-w_k,\tfrac{b}{2}\cdot h(R_2)-w_k\Bigr\}
& \geq \tfrac{2-b}{5-2b}\cdot w_k+\tfrac{1-b}{5-2b}\cdot\bigl(h(R_1)-w_k\bigr) \\
& \qquad \qquad +\tfrac{2}{5-2b}\cdot\bigl(\tfrac{b}{2}\cdot h(R_2)-w_k\bigr) \\
& =\frac{1}{5-2b}\Bigl(ah(R_1)+bh(R_2)-w_k\Bigr) \\
& =\frac{\UB}{5-2b}\geq\UB/5\geq\OPT/5. \qedhere
\end{align*}
\end{proof}

\begin{theorem} \label{5apxthm}
There is a $5$-approximation algorithm for MST interdiction. 
\end{theorem}

\begin{proof} 
If Theorem~\ref{bipoint} returns an optimal solution, we are done. Otherwise, we
return the best among a min-cut of $(V,E_{\leq k-1})$ (which has value at least $w_k$),
the set $R_1$, and the interdiction set returned by Lemma~\ref{tkncor}. The proof now
follows from Claim~\ref{calcclm}.  
\end{proof}

\subsection{Improvement to the guarantee stated in Theorem~\ref{mstinterthm}} 
\label{4apx} 

The improved approximation guarantee of $4$ comes from the fact that instead of focusing only
on $R_2$, 
we now {\em interpolate} between $R_1$ and $R_2$ to obtain our interdiction set $R$,
i.e., we return $R$ such that $R_1\sse R\sse R_2$. Since we always include $R_1$, we change
the definition of the tree-knapsack instance that we create accordingly. The tree $\Gm$
and the node weights $\{\al_q\}$ are unchanged; the weight of $v^{A,i}$ is now 
$\beta^\new_{v^{A,i}}:=c\bigl(R^\new(v^{A,i})\bigr)$, where 
$R^\new(v^{A,i}):=R(v^{A,i})\sm R_1=\bigl(\dt(A)\sm R_1\bigr)\cap E_i$,
and our budget is $B^\new:=B-c(R_1)$. For $N'\sse N$, define
$R^\new(N'):=R_1\cup\bigcup_{q\in N'}R^\new(q)$. Observe that $R^\new(N) \subseteq R_2$.

Since $R_1\sse R_2$, each component $U$ of $(V,E_{\leq i}\sm R_1)$ is a union of
components of $(V,E_{\leq i}\sm R_2)$, and hence, maps to a subset $S$ of the nodes of
$\Gm$ at depth $k-i$. We exploit the fact that since we include $R_1$ in our 
interdiction set, if we pick $\ell$ nodes from $S$, then we create $\min\{\ell+1,|S|\}$
components within $U$; this $+1$ term that we accrue (roughly speaking) from all
components of $(V,E_{\leq j}\sm R)$ over all $j=0,\ldots,k-1$ is the source of our
improvement. 

The following variant of Corollary~\ref{strongtkn} exploits the structure of the
tree-knapsack instance obtained from the MST-interdiction problem, which we then utilize
to obtain an interdiction set with an improved bound on $h(R)$ (Lemma~\ref{newinter}).

\begin{lemma} \label{newtkn}
Let $\bigl(\Gm,\{\al_v\},\{\beta_v\},B\bigr)$ be an instance of the tree knapsack problem
such that $\al_v=\al^{(i)}$ for all $v\in\level_i(\Gm)$ and all $i\geq 1$.
Let $\prt_i$ be a partition of $\level_i(\Gm)$ for all $i\geq 1$.
Let $\tht\in[0,1]$ be such that $(\hx_q=\tht)_{q\in N}$ is a feasible solution to
\eqref{tknlp}. 
We can obtain in polytime an integer solution $\tx$ to \eqref{tknlp} such that
$$
\sum_{i\geq 1}\sum_{S\in\prt_i}\al^{(i)}\min\{\tx(S)+1,|S|\}
\geq\sum_{i\geq 1}\al^{(i)}|\level_i|\tht
+\sum_{i\geq 1:|\prt_i|>1}\al^{(i)}\Bigl((1-\tht)|\prt_i|-1\Bigr).
$$
\end{lemma}

\begin{proof} 
Define $g(x):=\sum_{i\geq 1}\sum_{S\in\prt_i}\al^{(i)}\min\{x(S)+1,|S|\}$.
As usual $\level_i$ denotes $\level_i(\Gm)$.
Define $\al'\in\R_+^N$ as follows. For each $i$ such that $|\prt_i|=1$, set $\al'_v=\al_v$
for all $v\in\level_i$. For each $i$ with $|\prt_i|>1$ and each $S\in\prt_i$,
pick some node $v_S\in S$; set $\al'_{v_S}=0$ and $\al'_v=\al_v$ for all 
$v\in S\sm\{v_S\}$. 
We claim that for $\tx\in\{0,1\}^{N}$, we have 
\begin{equation}
g(\tx)\geq\sum_{v\in N}\al'_v\tx_v
+\sum_{\substack{i\geq 1: |\prt_i|=1, \\ \tx(\level_i)<|\level_i|}}\al^{(i)}
+\sum_{i\geq 1:|\prt_i|>1}\al^{(i)}|\prt_i| \ . \label{newtknineq1}
\end{equation}
To see this, consider any level $i\geq 1$. If $|\prt_i|=1$, the total contribution from
this level to $g(\tx)$ is $\al^{(i)}\min\{\tx(\level_i)+1,|\level_i|\}$, which is the same
as the contribution from this level to the RHS of \eqref{newtknineq1}.
If $|\prt_i|>1$, consider each set $S\in\prt_i$. The contribution from $S$ to $g(\tx)$ is  
$\al^{(i)}\min\{\tx(S)+1,|S|\}$, and the contribution from $S$ to the RHS of
\eqref{newtknineq1} is $\al^{(i)}\bigl(\tx(S\sm\{v_S\})+1\bigr)$, which is no larger.

To complete the proof, note that by Corollary~\ref{strongtkn}, we can obtain an integer
solution $\tx$ to \eqref{tknlp} such that 
\begin{equation*}
\begin{split}
\sum_{v\in N}\al'_v\tx_v & \geq 
\Bigl(\max_{\text{feasible solutions $x$ to \eqref{tknlp}}}\sum_{v\in N}\al'_vx_v\Bigr)
-\max_{\text{chains $C\sse N$}}\sum_{i\geq 1:\tx(\level_i)<|\level_i|}\al'(C_i) \\
& \geq \sum_{v\in N}\al'_v\hx_v-\sum_{i\geq 1:\tx(\level_i)<|\level_i|}\al^{(i)} \\
& = \sum_{i\geq 1}\al^{(i)}|\level_i|\tht-\sum_{i\geq 1:|\prt_i|>1}\al^{(i)}|\prt_i|\tht
-\sum_{i\geq 1:\tx(\level_i)<|\level_i|}\al^{(i)}.
\end{split}
\end{equation*}
Therefore, 
\begin{align*}
g(\tx) & \geq \sum_{v\in N}\al'_v\tx_v
+\sum_{\substack{i\geq 1: |\prt_i|=1, \\ \tx(\level_i)<|\level_i|}}\al^{(i)}
+\sum_{i\geq 1:|\prt_i|>1}\al^{(i)}|\prt_i| \\
& \geq \sum_{i\geq 1}\al^{(i)}|\level_i|\tht
+\sum_{i\geq 1:|\prt_i|>1}\al^{(i)}|\prt_i|(1-\tht)
-\sum_{\substack{i\geq 1: |\prt_i|>1, \\ \tx(\level_i)<|\level_i|}}\al^{(i)} \\
& \geq \sum_{i\geq 1}\al^{(i)}|\level_i|\tht
+\sum_{i\geq 1:|\prt_i|>1}\al^{(i)}\Bigl((1-\tht)|\prt_i|-1\Bigr). \qedhere
\end{align*}
\end{proof}

\begin{lemma} \label{newinter}
We can obtain a feasible interdiction set $R$ such that  
$h(R)\geq\frac{a}{2}\cdot h(R_1)+\frac{b}{2}\cdot h(R_2)-\frac{a}{2}\cdot w_k$.
\end{lemma}

\begin{proof} 
We apply Lemma~\ref{newtkn} to the tree-knapsack instance given
by $\bigl(\Gm,\{\al_q\},\{\beta^\new_q\},B^\new\bigr)$, but
we need to specify the partitions $\prt_i$ for all levels $i\geq 1$, and the value
$\tht$. 

For $i=0,\ldots,k$, let ${\mathcal{B}}_i$ denote the partition of $V$ induced by the
connected components of the multigraph $(V, E_{\le i} \setminus R_1)$. Since 
$R_1\sse R_2$, the partition $\A_i$ refines (not necessarily strictly) the partition
$\B_i$ for all $i=0,\ldots,k$.
The components in $\B_{k-i}$ therefore naturally induce a partition $\prt_i$ of the nodes of $\Gm$ at depth $i$, consisting of the sets 
$\{v^{A,k - i}: A\in \mathcal{A}_{k-i}, A\subseteq S\}_{S\in\B_{k-i}}$. 

We apply Lemma~\ref{newtkn} to the tree-knapsack instance
$\bigl(\Gm,\{\al_q\},\{\beta^\new_q\},B^\new\bigr)$, taking
$\al^{(i)}=w_{k-i+1}-w_{k-i}$ and $\prt_i$ to be the partition defined above, for all
$i=1,\ldots,k$, and $\tht=\frac{b}{2}$. 
We need to show that $\hx:=(\hx_q=\tht)_{q \in N}$ is a feasible solution to \eqref{tknlp} for this
tree-knapsack instance. This follows because 
$\beta^\new(N) = 2c\bigl(\bigcup_{q\in N}R^\new(q)\bigr) \le 2\bigl(c(R_2)-c(R_1)\bigr)$  
and we have $(1-b)\cdot c(R_1)+b\cdot c(R_2)=B$, so 
$$
\sum_q\beta^\new_q\hx_q=\tht\beta^\new(N)\leq b\bigl(c(R_2)-c(R_1)\bigr)=B-c(R_1)=B^\new.
$$

Let $\tx$ be the integer solution returned by Lemma~\ref{newtkn}, which specifies a
downwards-closed set $N'\sse N$. Let $R=R^\new(N')$. 
We first show that, analogous to Lemma~\ref{lem:translate}, $R$ is feasible, and
$h(R)\geq g(\tx):=\sum_{i\geq 1}\sum_{S\in\prt_i}\al^{(i)}\min\{\tx(S)+1,|S|\}$.
We have 
\begin{equation*}
\begin{split}
c(R)=c(R_1)+c\Bigl(\bigcup_{q\in N'}R^\new(q)\Bigr)
& \leq c(R_1)+\sum_{q\in N'} c\bigl(R^\new(q)\bigr) \\
& = c(R_1)+\beta^\new(N')\leq c(R_1)+B^\new = B.
\end{split}
\end{equation*}
Consider any index $0 \le i \le k - 1$.
As in the proof of part (ii) of Lemma \ref{lem:translate}, for every node $v^{A, i} \in N'$, we 
know that $A$ is a component of $(V, E_{\le i}\setminus R)$. 
Consider any $S\in\prt_{k-i}$, and let 
$U=\bigcup_{v^{A,i}\in S}A$. Note that if $S\sm N'\neq\es$, then $\bigcup_{v^{A,i}\in S\sm N'}A$ is non-empty. 
So there are always at least $\min\bigl\{|N' \cap S| + 1, |S|\bigr\}$ components of 
$(V,E_{\leq i}\sm R)$ contained in $U$.
Therefore, by Lemma~\ref{lem:val} (and since $\prt_i$ is a partition of $\level_i$ for
each $i$), we obtain 
$$
h(R)=\val(R)+w_k
\ge \sum_{i = 0}^{k-1}\sum_{S \in \prt_{k-i}}(w_{i+1} - w_i)\min\bigl\{|N' \cap S| + 1, |S|\bigr\}
=g(\tx).
$$
The guarantee in Lemma~\ref{newtkn} then yields the following. Recall that $a=1-b$. 
\begin{align}
h(R) & \geq\sum_{i=0}^{k-1}(w_{i+1}-w_i)\sg(E_{\leq i}\sm R_2)\cdot\tfrac{b}{2}
+\sum_{\substack{i=0,\ldots,k-1: \\ \sg(E_{\leq i}\sm R_1)>1}}
(w_{i+1}-w_i)\biggl[\Bigl(1-\tfrac{b}{2}\Bigr)\sg(E_{\leq i}\sm R_1)-1\biggr] \notag \\
& \geq\tfrac{b}{2}\cdot h(R_2)
+\sum_{\substack{i=0,\ldots,k-1: \\ \sg(E_{\leq i}\sm R_1)>1}}
(w_{i+1}-w_i)\sg(E_{\leq i}\sm R_1)\cdot\tfrac{a}{2} \label{4apxineq1} \\
& =\tfrac{b}{2}\cdot h(R_2)
+\sum_{i=0}^{k-1}(w_{i+1}-w_i)\sg(E_{\leq i}\sm R_1)\cdot\tfrac{a}{2}
-\sum_{\substack{i=0,\ldots,k-1: \\ \sg(E_{\leq i}\sm R_1)=1}}(w_{i+1}-w_i)\cdot\tfrac{a}{2} \notag \\
& \geq\tfrac{a}{2}\cdot h(R_1)+\tfrac{b}{2}\cdot h(R_2)-\tfrac{a}{2}\cdot w_k. \notag
\end{align}
Inequality \eqref{4apxineq1} follows since $t\bigl(1-\frac{b}{2}\bigr)-1\geq t(1-b)/2$ for
all $t\geq 2$.
\end{proof} 

\begin{proofof}{Theorem~\ref{mstinterthm}}
We either return an optimal solution found by Theorem~\ref{bipoint}, or return
the better of a min-cut of $(V,E_{\leq k-1})$ and the interdiction set returned by 
Lemma~\ref{newinter}. We obtain a solution of value 
\begin{align*}
\max\Bigl\{w_k,\tfrac{a}{2}\cdot h(R_1)&+\tfrac{b}{2}\cdot h(R_2)-\bigl(1+\tfrac{a}{2}\bigr)w_k\Bigr\}
\\ & \geq \frac{1+a}{3+a}\cdot w_k
+\frac{2}{3+a}\cdot\Bigl(\tfrac{a}{2}\cdot h(R_1)+\tfrac{b}{2}\cdot h(R_2)-\bigl(1+\tfrac{a}{2}\bigr)w_k\Bigr)
\\
& = \frac{ah(R_1)+bh(R_2)-w_k}{3+a}=\frac{\UB}{3+a}\geq\UB/4\geq\OPT/4. \qedhere
\end{align*}
\end{proofof}

\subsection{Lower bound on the approximation ratio achievable relative to \boldmath \UB} 
\label{lbound}  

We show that for every $\e>0$, there exist MST-interdiction instances, where 
$\UB/\OPT\geq 3-\e$. This implies that one cannot achieve an approximation ratio better
than $3$ when comparing against the upper bound $\UB$ used in our analysis (and the one
in~\cite{Zenklusen15}).   

\begin{theorem} \label{theo:integrality_gap}
For any fixed $\e>0$, there exists an instance of MST interdiction where 
$\UB/\OPT\geq 3-\e$. 
\end{theorem}

\begin{proof}
Our instance is a graph $G = (V,E)$, where $V := \{v_1, \dots, v_n\}$ with
$n\ge\min\{4,4/\e\}$. 
The edge set is $E=E_1\cup E_2$, where 
$E_1:=\{v_1v_2, v_2v_3,\dots,v_{n-2}v_{n-1},v_{n-1}v_1\}$ is a simple cycle on
$v_1,\ldots,v_{n-1}$, and $E_2:=\{v_1v_n, v_2v_n,\dots, v_{n-1}v_n\}$ is a  star
rooted at $v_n$ with leaves $v_1,\ldots,v_{n-1}$.   
The edges in $E_1$ have weight $w_1=0$ and interdiction cost $n$, while
the edges in $E_2$ have weight $w_2=1$ and interdiction cost $2n$. 
The interdiction budget is $B=2n-2$. 

Observe that the quantity $k$, as defined in Claim \ref{ass:connected}, is equal to $2$.
Taking $R=\es$, the graph $(V, E_{\le 1}\sm R)$ is disconnected, so $k\geq 2$. 
Any feasible interdiction set $R$ contains at most one edge from $E_1$
and no edges from $E_2$, so $(V, E_{\le 2}\sm R)$ is connected, and therefore $k=2$.

This also implies that $\val(R)\leq 1$ for any feasible interdiction set $R$: since 
$R\sse E_1$ and  $|R\cap E_1|\leq 1$, we can construct a spanning tree of $G-R$ by taking
$n-2$ edges from $E_1 \setminus R$ and any edge from $E_2$. So $\OPT=1$.  

Now we proceed to compute the upper bound $\UB$. For $R\sse E_{\leq 1}$, we have
$f_\ld(R)=1$ if $R=\es$, and $|R|(1-n\ld)$ otherwise. 
Therefore, 
$$
\eta(\ld):=\ld B+\max_{R\sse E_{\leq 1}}f_\ld(R)=\ld B+\max\bigl\{1,(n-1)(1-n\ld)\bigr\}
=\max\Bigl\{\ld B+1,(n-1)-\ld\bigl(n(n-1)-B\bigr)\Bigr\},
$$
which is minimized at $\ld=\frac{n-2}{n(n-1)}$. Therefore 
$\UB:=\min_{\ld\geq 0}\eta(\ld)=\frac{2(n-2)}{n}+1=3-\frac{4}{n}\geq (3-\e)\OPT$.
\end{proof}

\section{Extension to metric-TSP interdiction} \label{tspinter}
In the {\em metric-TSP interdiction} problem, we are given a complete graph $G = (V, E)$
with metric edge weights 
$\{w_e\}_{e \in E}$ and nonnegative interdiction costs $\{c_e\}_{e \in E}$, along with a
nonnegative budget $B$. The goal is to find a set of edges $R \subseteq E$ such that 
$c(R)\le B$ so as to maximize the 
minimum $w$-weight of a closed walk in the graph $G - R$ that visits each vertex at least
once. Zenklusen \cite{Zenklusen15} observed that an $\alpha$-approximation
algorithm for the MST interdiction problem yields a $2\alpha$-approximation algorithm
for the metric-TSP interdiction problem. As a corollary to our Theorem
\ref{theo:MST_interdiction}, we therefore obtain the following result. 

\begin{theorem} \label{theo:TSP_interdiction}
There is an $8$-approximation algorithm for the metric-TSP interdiction problem.
\end{theorem}

\section{Maximum-spanning-tree interdiction} \label{maxstinter}
We now consider the {\em maximum-spanning-tree interdiction} problem, wherein the input
$\bigl(G=(V,E),\{w_e\geq 0\}_{e\in E},\{c_e\geq 0\}_{e\in E},B\bigr)$ is the same as in
the MST interdiction problem, but the goal is
to remove a set $R\sse E$ of edges with $c(R)\leq B$ so as to {\em minimize} the
$w$-weight of a {\em maximum spanning tree} of $G-R$. 
We show that this problem is at least as hard as the minimization version of the 
{\em densest-$k$-subgraph} problem (\mindks), wherein we seek a minimum-size set $S$ of
nodes in a given graph such that at least $k$ edges have both endpoints in $S$.
This shows a stark contrast between MST interdiction and maximum-spanning-tree (\maxst)
interdiction. 

\begin{theorem}
An $\al(m,n)$-approximation algorithm for the maximum-spanning-tree interdiction 
problem for instances with $m$ edges, $n$ nodes, yields a
$2\alpha(m+n-1,n)$-approximation algorithm for \mindks for instances 
{with $m$ edges and $n$ nodes.}  
\end{theorem}

\newcommand{\IMaxST}{\ensuremath{\mathcal{I}'}}
\newcommand{\IDS}{\ensuremath{\mathcal{I}}}

\begin{proof}
Let $\IDS=\bigl(H=(N,F),k\bigr)$ be a \mindks instance, with $|N|=n$, $|F|=m$. 
We may assume that $|F| \ge k$ as otherwise the instance is infeasible.
We construct the following instance $\IMaxST$ of the maximum-spanning-tree interdiction
problem. The underlying multigraph is $G=(N,E:=E' \cup F)$, where $E'$ is an arbitrary tree spanning $N$. Define $w_e = 0$ for all $e \in E'$, and $w_e = 1$ for all $e \in F$. 
The interdiction costs are $c_e = 1$ for all $e \in F$, and 
$c_e=m-k+1$ for all $e\in E'$. Finally, we set the budget to $B=m-k$. Thus, if $R$ is a
feasible interdiction set, we must have $R\sse F$, and so $G-R$ is connected and the
interdiction problem has a finite optimal value.

We show that: (1) if $R\sse F$ is a feasible interdiction set, then the set $S$ of
non-isolated nodes of $(N,F\sm R)$ is a feasible \mindks solution of value at most
$2\cdot\maxst(G-R)$, where $\maxst(G-R)$ is the weight of a maximum spanning tree of
$G-R$;  
(2) conversely, if $S \subseteq N$ is a feasible \mindks solution, then $F\sm F(S)$ is a feasible
interdiction set with objective value at most $|S|$, where $F(S)$ is the set of edges in $F$
having both endpoints in $S$. 

These two statements imply the theorem as follows. 
Let $\A$ be the stated $\al=\al(m+n - 1,n)$-approximation algorithm for
maximum-spanning-tree interdiction. 
We run $\mathcal{A}$ to obtain a feasible interdiction set $R$, which yields a
corresponding \mindks solution $S$. Then,
$$
|S|\leq 2\cdot\maxst(G-R)\leq 2\al\OPT(\IMaxST)\leq 2\al\OPT(\IDS) \ ,
$$
where the first and last inequalities follow from statements (1) and (2) above.

We now prove statements (1) and (2).
Let $R\sse F$ be such that $c(R)=|R|\leq B$.
Let $S$ denote the set of non-isolated vertices in the graph $(N, F \setminus R)$, so
every node in $S$ has at least one edge of $F\sm R$ incident to it. 
First, we argue that $S$ is a feasible \mindks-solution. 
Since each vertex of $N \setminus S$ is isolated in the graph $(N,F \setminus R)$, it
follows that $R\supseteq F\sm F(S)$.
Therefore, $|F|-|F(S)|\leq|R|\leq B=m-k$, and so $|F(S)|\geq k$.
The weight of a maximum spanning tree in $G - R$ is equal to $|S| - \sigma$, where
$\sigma$ is the number of connected components of the graph $(S, F\setminus R)$. By the
definition of $S$, this multigraph has no isolated vertices. So $\sigma \le |S| / 2$, and
therefore $\maxst(G-R)= |S|-\sigma \ge |S|/2$. This proves (1).

Conversely, suppose $S\sse N$ is such that $|F(S)|\geq k$. Then $R=F\sm F(S)$ satisfies
$c(R)=m-|F(S)|\leq B$, so is a feasible interdiction set. We have 
$\maxst(G-R)=|S|-\sg\leq |S|$,
where $\sg$ is the number of connected components of $(S,F\sm R)$. This proves (2).
\end{proof}

The above hardness result continues to hold with unit interdiction costs, since we can
replace each edge $e$ with $c_e=m-k+1$ in the above reduction with $m-k+1$ parallel  
unit-cost edges (of weight 0).  
Our reduction creates a \maxst-interdiction instance with only two distinct edge weights
$w_1<w_2$. This interdiction problem can be seen as a special case of the following 
{\em matroid interdiction} problem (involving the graphic matroid on the ground set
$\{e\in E: w_e=w_2\}$): given a matroid with ground set $U$ and rank function 
$\rk$, interdiction costs $c:U\mapsto\R_+$, and budget $B$, minimize $\rk(U\sm R)$ subject
to $c(R)\leq B$. 
Our hardness result for \maxst interdiction thus also implies that matroid interdiction is
\mindks-hard. A related rank-reduction problem---minimize $c(R)$ subject to
$\rk(U\sm R)\leq\rk(U)-k$---was considered by~\cite{JoretV15} and shown to be 
\mindks-hard for transversal matroids (but not for graphic matroids, wherein this is
essentially the min $k$-cut problem).      

We remark that it is possible to achieve {\em bicriteria} approximation guarantees for
\maxst interdiction: we can obtain a solution of weight $W\leq(1+\e)\OPT$ 
while violating the budget by a $\bigl(1+\frac{1}{\e}\bigr)$ factor (and $W>\OPT$ implies 
no budget violation). This follows by taking $\ld=\e\OPT/B$  
in the Lagrangian problem $\min_R\bigl(\maxst(G-R)+\ld c(R)\bigr)$, which is a submodular
minimization problem that can be solved exactly; it also follows from the work
of~\cite{ChestnutZ15b}.

\bibliographystyle{plain}

\appendix

\section{Proof of Theorem~\ref{bipoint}} \label{append-prelim}
For any $\ld\geq 0$, by supermodularity of $f_\ld$, there is a unique minimal set in
$\Oc^*_\ld$, which we denote by $R^*_\ld$. We remark that $R^*_\ld$ can be computed in
polytime using an algorithm for submodular-function minimization (see,
e.g.,~\cite{FleischerI03}). Let $g(\ld)$ denote the optimal value of \eqref{lagp}. Then
$g(\ld)$ is the maximum of a finite collection of nonincreasing linear functions, so it is a continuous, 
piecewise-linear, nonincreasing, convex function. Further, $\ld\geq 0$ is a {\em breakpoint}
of $g(.)$ iff there are at least two distinct values in $\{c(R): R\in\Oc^*_\ld\}$. Also, if
$\ld$ is not a breakpoint, then the slope of $g$ at $\ld$ is $-c(R^*_\ld)$.

We first prove the following: 
\begin{equation}
\text{if $\lambda > \lambda'$, \ \ then \ \ $R^*_\ld \subseteq R^*_{\ld'}$}. \tag{*} 
\label{nested}
\end{equation}

\smallskip \noindent
Let $R=R^*_\ld$ and $R'=R^*_{\ld'}$. 
Suppose that $R\not\subseteq R'$. We have 
$f_\ld(R\cap R')+f_\ld(R\cup R')\geq f_\ld(R)+f_\ld(R')$. Since $R\cap R'\subsetneq R$
(since we assume that $R'$ is not a superset of $R$), and 
$R=R^*_\ld$ is the minimal set in $\Oc^*_\ld$, we have $f_\ld(R\cap R')<f_\ld(R)$. So we must
have $f_\ld(R\cup R')>f_\ld(R')$. But then 
$$
f_{\ld'}(R\cup R')=f_\ld(R\cup R')+(\ld-\ld')c(R\cup R')>f_\ld(R')+(\ld-\ld')c(R')=f_{\ld'}(R') \ ,
$$
which contradicts that $R'\in\Oc^*_{\ld'}$.

Next we present the two approaches for proving the theorem. The first one utilizes
binary search and yields a more elementary, but weakly polytime algorithm. The second one
utilizes the fact that $g(\ld)$ has polynomially many breakpoints and slopes, and uses
results on parametric submodular-function minimization to obtain $g(\ld)$ in strongly
polynomial time, which then yields the theorem.

If $c(R^*_0)\leq B$, then $\val(R^*_0)\geq\OPT$ and we are done, so assume this is not the
case in the sequel. Let $M$ be the smallest integer such that all the $w_e$s, $c_e$s, and
$B$ are multiples of $\frac{1}{M}$; note that $\log M$ is polynomially bounded. For
$\ld=K:=2w_k(n-1)M$, we must have $c(R^*_\ld)\leq B$ as otherwise 
$c(R^*_\ld)\geq B+\frac{1}{M}$, and so 
$\val(R^*_\ld)-\ld c(R^*_\ld)\leq -w_k(n-1)-\ld B<\OPT-\ld B$. 

\paragraph{Binary search.}
As noted above $c(R^*_K)\leq B$. Let $\e:=\frac{1}{M^2c(E)^2}$.
We perform binary search in the interval $[0,K]$ to find $\ld_1,\ld_2\in[0,K]$ with
$0 < \ld_1-\ld_2\leq\e/2$ such that $c(R^*_{\ld_1})\leq B<c(R^*_{\ld_2})$. If 
$c(R^*_{\ld_1})=B$, then we have $\val(R^*_{\ld_1})\geq\OPT$ and we are done; so assume
this does not happen.
Let $R_1:=R^*_{\ld_1}$ and $R_2 := R^*_{\ld_2}$. Note that since $\lambda_1 > \lambda_2$,
we have $R_1 \subseteq R_2$ by \eqref{nested}. The only thing left to prove is that
$R_1,R_2\in\Oc^*_\ld$ for some $\ld\geq 0$.

We claim that any two breakpoints of $g$ are separated by at
least $\e$. This is because if $\ld$ is a breakpoint and $A, B\in\Oc^*_\ld$ are such that
$c(A)>c(B)$, then we have
$\ld=\frac{\val(A)-\val(B)}{c(A)-c(B)}$, which can be written as a fraction with integer
numerator and positive integer denominator bounded by $Mc(E)$. 
So the interval $[\ld_2,\ld_1]$ contains at most one breakpoint, and hence, exactly one
breakpoint (since $c(R_1)\neq c(R_2)$). Let $\ld$ be this breakpoint. Then, $g$ is linear
in $(\ld_2,\ld)$ and $(\ld,\ld_1)$ with slopes $-c(R_2)$ and $-c(R_1)$ respectively, so
$g(\ld)=\val(R_1)-\ld c(R_1)=\val(R_2)-\ld c(R_2)$. Thus, $R_1,R_2\in\Oc^*_\ld$. 

\paragraph{Parametric submodular-function minimization.}
By \eqref{nested}, we know that $\{R^*_{\ld}: \ld\geq 0\}$ is a nested family, and hence
consists of at most $|E|+1$ sets. Thus, $g$ consists of at most $|E|+1$ linear
segments. Since $-f_\ld$ is submodular, one can use algorithms for parametric
submodular-function minimization~\cite{FleischerI03,Nagano07} to obtain
the slopes of all these segments (and the corresponding sets in $\Oc^*_\ld$) in strongly
polynomial time. Now if some slope is equal to $-B$, then the corresponding interdiction
set is an optimal solution to the MST-interdiction problem.  
Otherwise, we have that the slope of $g$ at $\ld=0$ is less than $-B$, and the
slope at $\ld=K$ is more than $-B$, so there is some breakpoint $\ld$ where $g$ has slope
less than $-B$ at $\ld^-$ (a value infinitesimally smaller than $\ld$), and more than $-B$
at $\ld^+$ (a value infinitesimally larger than $\ld$). Then, the interdiction sets
$R_1$ and $R_2$ corresponding to the slopes at $\ld^+$ and $\ld^-$ respectively satisfy
the theorem. \hfill \qed 

\end{document}